\documentclass[conference]{IEEEtran}

\IEEEoverridecommandlockouts
\def\BibTeX{{\rm B\kern-.05em{\sc i\kern-.025em b}\kern-.08em
    T\kern-.1667em\lower.7ex\hbox{E}\kern-.125emX}}

\usepackage{graphicx}
\usepackage{url}
\usepackage{algorithm}
\usepackage{algpseudocode}
\usepackage{float}
\usepackage{amssymb}
\usepackage{amsmath}

\usepackage{amsthm}    
\newtheorem{thm}{Theorem}
\newtheorem{mydefn}{Definition}

\newtheorem{lem}{Lemma}
\newtheorem{prob}{Problem}

\newtheorem{reform}{Reformulation}
\newtheorem{assm}{Assumption}

\makeatletter
\let\NAT@parse\undefined
\makeatother
\usepackage{hyperref}
\hypersetup{
    colorlinks=true,      
    linkcolor=black,
    citecolor=black,
    filecolor=black,
    urlcolor=black     
    }

\newcommand{\R}{\mathbb{R}}
\newcommand{\N}{\mathbb{N}}
\renewcommand{\P}{\mathbb{P}}

\newcommand{\pr}[1]{\mathbb{P}\left(#1\right)}
\newcommand{\ex}[1]{\mathbb{E}\left[#1\right]}

\newcommand{\var}[1]{\mathrm{Var}\left(#1\right)}
\newcommand{\std}[1]{\mathrm{Std}\left(#1\right)}

\newcommand{\bvec}[1]{\vec{\boldsymbol{#1}}}
\newcommand{\Nt}[2]{\mathbb{N}_{[#1,#2]}}

\author{Shawn Priore and Meeko Oishi
    \thanks{
        This material is based upon work supported by the National Science Foundation under NSF Grant Numbers CMMI-2105631.  Any opinions, findings, and conclusions or recommendations expressed in this material are those of the authors and do not necessarily reflect the views of the National Science Foundation.  
        \newline \indent Shawn Priore and Meeko Oishi are with the Department of Electrical and Computer Engineering, University of New Mexico, Albuquerque, NM; e-mail: \texttt{shawnpriore@unm.edu} (corresponding author)\texttt{, oishi@unm.edu}.
    }
}
\title{Chance Constrained Stochastic Optimal Control for Linear Systems with a Time Varying Random Control Matrix}

\begin{document}

\maketitle

\begin{abstract}
This work proposes an open loop methodology to solve chance constrained stochastic optimal control problems for linear systems with a stochastic control matrix. We consider a joint chance constraint for polytopic time-varying target sets under moment and unimodality assumptions. We reformulate the chance constraint into individual biconvex constraints using the one-sided Vysochanskij–Petunin inequality. We demonstrate our methodology on two spacecraft rendezvous problems. We compare the proposed method with the scenario approach and moment-based methods based on Cantelli's inequality.
\end{abstract}

\section{Introduction}

It is well known in the satellite operations community that system modeling inaccuracies tend to result in underestimations of fuel consumption. Anecdotally, operators have had to allot up to an additional 50\% of the predicted fuel use to account for these inaccuracies. In many cases, these inaccuracies arise when dynamic systems are modeled under the assumption that the control matrix is perfectly known and the input is perfectly executed. In practice, inaccuracies in the control can stem from under-performing actuators, non-impulsive transitions between control inputs, or time delays in sensing-control algorithms. Further, the level of inaccuracy can vary over time. As is the case when a time delay in the sensing-control algorithms is caused by process queuing in the processor as the number of processes queued prior to the algorithm can be random. As inaccuracies like these can be challenging to model, we propose modeling these inaccuracies as stochastic elements in the control matrix. By adding stochasticity in the control matrix, we can model inaccuracies that result from control actions in a manner that is not sensitive to the actual cause of the stochasticity. 

Stochastic control matrices present challenges that have not been well addressed in previous literature. Previous work has considered the presence of random time delays between sensing and control actuation as in \cite{Branicky2000, Huo2005}. The introduction of a random time delay can be particularly useful in scenarios where impulsive control assumptions are impractical. However, the methods rely on filtering techniques that don't always allow for probabilistic assurances. Robust model predictive control methods have been posed for linear parameter varying systems where bounded random parameters are present in the control and state matrices \cite{Casavola2000, Calafiore2011, Gravell2021}. Under the assumption of boundedness, robust models can be used to derive stabilization results and optimal feedback controllers \cite{Mahmoud1999, Calafiore2012}. The reliance on boundedness can limit the applicability of the method, particularly when bounds are uncertain. Both of these methodologies focus solely on a particular type of uncertainty and the proposed solutions lack generalizable results. Little work has been presented to address stochastic control matrices outside of linear parameter varying and time delay systems.

With how scarce fuel is for on orbit satellites, techniques that can better account for inaccuracies and better evaluate fuel consumption needs can be invaluable tools for operators. To address concerns of limited resources and random elements in the control matrix, we consider this problem in a stochastic optimal control framework. Within the stochastic optimal control framework, we consider probabilistic joint chance constraints for time-varying polytopic target sets while minimizing a convex cost function such as fuel usage. To solve the stochastic optimal control problem, our approach employs a biconvex approximation of the polytopic target set chance constraints where solutions can be found via convex optimization techniques. We rely on Boole's inequality \cite{casella2002} to bound the joint chance constraint by a series of individual chance constraints. Then the individual chance constraints are reformulated into biconvex constraints with the one-sided Vysochanskij–Petunin inequality \cite{Mercadier2021} under assumptions of known moments and unimodality. The one-sided Vysochanskij–Petunin inequality is a tightening of Cantelli's inequality \cite{casella2002} for unimodal distributions and is currently the tightest moment based bound available for unimodal constraints. So, while the one-sided Vysochanskij–Petunin inequality will introduce conservatism to the solution, it will also enable optimization under a wide range of stochastic assumptions. To address biconvexity in our reformulations, we discuss the alternate convex search method \cite{Leeuw1994} to find solutions. {\em The main contribution of this paper is the construction of a tractable optimization problem that solves for convex joint chance constraints in the presence of random elements in the control matrix.} 

The paper is organized as follows. Section \ref{sec:prelim} provides mathematical preliminaries and formulates the optimization problem. Section \ref{sec:methods} derives the reformulation of the chance constraints with Boole's inequality and the one-sided Vysochanskij–Petunin inequality. Section \ref{sec:results} demonstrates our approach on two spacecraft rendezvous problems, and Section \ref{sec:conclusion} provides concluding remarks.

\section{Preliminary and Problem Setup} \label{sec:prelim}

We denote random variables with bold case, $\bvec{x}$ or $\boldsymbol{x}$, regardless of dimension. For a random variable $\boldsymbol{x}$, we denote the expectation as $\ex{\boldsymbol{x}}$, standard deviation as $\std{\boldsymbol{x}}$, and variance as $\var{\boldsymbol{x}}$. For a matrix $A$, the operator $\mathrm{vec}(A)$ vertically concatenates the columns of $A$ into a column vector. For two matrices $A$ and $B$, we denote the Kronecker product as $A \otimes B$. For matrix entries $A_1, \ldots, A_m$, we denote a  block diagonal matrix constructed with these elements as $\mathrm{blkdiag}(A_1, \ldots, A_m)$. We denote an identity matrix of size $n$ as $I_n$ and a $n$-dimensional vector of ones as $\vec{1}_n$.

\subsection{Problem Formulation}

We consider a discrete-time linear system given by
\begin{equation} \label{eq:dynamics}
    \bvec{x}(k+1) = A \bvec{x}(k) + \boldsymbol{B}(k) \vec{u}(k) 
\end{equation}
with state $\bvec{x}(k) \in \mathcal{X} \subseteq \R^n$, input $\vec{u}(k) \in \mathcal{U} \subseteq \R^m$, and time index $k \in \Nt{0}{N}$. We assume initial conditions $\vec{x}(0)$ are known, and the admissible control set $\mathcal{U}$ is convex. The control matrix $\boldsymbol{B}(k)$ is time-varying and real valued with probability space $(\Omega, \mathcal{B}(\Omega), \P_{\boldsymbol{B}})$ with outcomes $\Omega$, Borel $\sigma$-algebra $\mathcal{B}(\Omega)$, and probability measure $\P_{\boldsymbol{B}}$ \cite{casella2002}. 

We write the dynamics at time $k$ as an affine combination of the initial condition and the concatenated control sequence,
\begin{equation} \label{eq:lin_dynamics}
    \bvec{x}(k) = A^k \vec{x}(0) + \mathcal{A}(k) \boldsymbol{\mathcal{C}}\vec{U}
\end{equation}
with 
\begin{subequations}
\begin{alignat}{2}
    \vec{U} =& 
    \begin{bmatrix} 
        \vec{u}(0)^\top & 
        \ldots & 
        \vec{u}(N-1)^\top 
    \end{bmatrix}^\top &&\in \mathcal{U}^{N} \\
    {\mathcal{A}}(k) = & \begin{bmatrix} A^{k-1} & \ldots & A & I_n &  0_{n \times (N\!-\!k\!-\!1)m} \end{bmatrix} && \in \R^{n \times Nn}  \\
    \boldsymbol{\mathcal{C}} = &\;  \mathrm{blkdiag}\left( \boldsymbol{B}(0), \ldots, \boldsymbol{B}(N-1) \right) && \in \R^{Nn \times Nm} 
\end{alignat}
\end{subequations}    

We seek to minimize a convex performance objective 
\begin{equation}
    J: \mathcal{X}^{N} \times \mathcal{U}^{N} \rightarrow \R
\end{equation}
such as fuel cost. We presume the state must stay within time varying polytopic sets, represented by the half-space inequalities $\vec{G}_{ki} \bvec{x}(k) \leq h_{ki}$,  with a desired likelihood
\begin{equation}\label{eq:constraint_t}
        \pr{\bigcap_{k=1}^N \bigcap_{i=1}^{q_k} \vec{G}_{ki} \bvec{x}(k) \leq h_{ki} }  \geq  1-\alpha
\end{equation}
where $q_k$ is the number of linear inequalities at time $k$. We presume convex, compact, and polytopic sets $ \left\{ \bvec{x}(k) \middle| \cap_{i=1}^{q_k} \vec{G}_{ki} \bvec{x}(k) \leq h_{ki} \right\} \subseteq \R^n$, and probabilistic violation threshold $\alpha <  1/6$. 

\begin{mydefn}[Unimodal Distribution \cite{Bertin1997}]\label{def:unimodal}
A unimodal distribution is a distribution whose cumulative distribution function is convex in the region $(-\infty, a)$ and concave in the region $(a, \infty)$ for some $a \in \R$.
\end{mydefn}
\begin{assm}\label{assm:finite_moment}
Each random element of $\boldsymbol{B}(k)$ has a finite expectation and variance. 
\end{assm}
\begin{assm}\label{assm:unimodal}
Each random variable, $\vec{G}_{ki} \bvec{x}(k)$, marginally follows a unimodal distribution. 
\end{assm}

\begin{mydefn}[Strong Unimodal Distribution \cite{Bertin1997}] \label{defn:strong_unimodal}
A strong unimodal distribution is one in which unimodality is preserved by convolution.
\end{mydefn}

Assumptions \ref{assm:finite_moment} and \ref{assm:unimodal} will guarantee a closed form reformulation of the polytopic set chance constraints by means of the one-sided Vysochanskij-Petunin inequality. Assumption \ref{assm:unimodal} is the more restrictive of the two assumptions as unimodality can be challenging to show analytically. In rare cases, unimodality can be verified analytically by properties of strong unimodality. One method to check for strong unimodality is to verify that the probability density function is log concave \cite{Ibragimov1956}. For example, beta random variables with both parameters greater than or equal to $1$, gamma random variables with shape parameter greater than or equal to $1$, Gaussian, Laplacian, and exponential distributions are all strong unimodal. As a result, any affine summation of these random variables will always be unimodal.

However, the easiest method to validate unimodality is by randomly sampling the random control matrix and plotting the solutions applied to the random samples. By computing the empirical cumulative distribution function of the chance constraint with a large enough sample size, one can validate unimodality by checking concavity and searching for a single inflection point with a series of affine functions.
 
Finally, we write the optimization problem we seek to solve,
\begin{subequations}\label{prob:big_prob_eq}
    \begin{align}
        \underset{\vec{U}}{\mathrm{minimize}} \quad & J\left(
        \bvec{x}(1), \ldots, \bvec{x}(N), \vec{U}\right)  \\
        \mathrm{subject\ to} \quad  & \vec{U} \in \mathcal{U}^N,  \\
        & \text{Dynamics } \eqref{eq:dynamics} \text{ with }
        \vec{x}(0)\\
        & \text{Probabilistic constraint \eqref{eq:constraint_t}} \label{prob:initial_eq_prob_constraints} 
    \end{align}
\end{subequations}

\begin{prob} \label{prob:1}
    Under Assumptions \ref{assm:finite_moment}-\ref{assm:unimodal}, solve the stochastic optimization problem \eqref{prob:big_prob_eq} with open loop control $\vec{U} \in \mathcal{U}^N$, and probabilistic violation threshold $\alpha$.
\end{prob}
The main challenge in solving Problem \ref{prob:1} is assuring \eqref{prob:initial_eq_prob_constraints}. The interaction of the random elements of the control matrix and the control input makes enforcing the constraints challenging. Even if closed form expressions exist for the joint chance constraint's probability, changes in the control input can change the shape and concavity of the distribution and expression of the constraint.  

\section{Methods} \label{sec:methods}

To solve Problem \ref{prob:1}, we reformulate the joint chance constraint \eqref{eq:constraint_t} into a tractable and closed form approximation that is amenable to convex optimization techniques. As we will show in Section \ref{ssec:reform}, we can use the one-sided Vysochanskij-Petunin inequality to rewrite our probabilistic constraints into constraints that are affine in the expectation and standard deviation of the half space constraint. In Section \ref{ssec:solve}, we show the reformulation elicits a biconvex constraint and discuss the alternating convex search \cite{Leeuw1994} approach to solving the constraints.

\subsection{The One-sided Vysochanskij–Petunin Inequality} \label{ssec:vp_i}

Here, we state the one-sided Vysochanskij-Petunin inequality for reference.

\begin{thm}[\cite{Mercadier2021}]
Let $\boldsymbol{x}$ be a real valued unimodal random variable with finite expectation $\ex{\boldsymbol{x}}$ and finite, non-zero standard deviation $\std{\boldsymbol{x}}$. Then, for $\lambda > \sqrt{5/3}$, 
\begin{equation} \label{eq:vp}
    \pr{ \boldsymbol{x} - \ex{\boldsymbol{x}}  \geq  \lambda \std{\boldsymbol{x}}} \leq \frac{4}{9(\lambda^2+1)}
\end{equation}
\end{thm}

The one-sided Vysochanskij-Petunin inequality provides an upper bound on the tail probability of a unimodal random variable's deviation from its mean. The one-sided Vysochanskij-Petunin inequality is a refinement of Cantelli’s inequality for unimodal distributions. The two-sided Vysochanskij–Petunin inequality \cite{Vysochanskij1980} is commonly cited as the foundation for the $3\sigma$ rule in statistics, and the one-sided Vysochanskij-Petunin inequality can be used to show a similar relationship for single-tail bounds.

Here, we will use \eqref{eq:vp} to bound the target set chance constraint probabilities. By moving the expectation to the right side of the inequality, we can bound the probabilities of random variables based on an affine combination of the expectation and standard deviation.

\subsection{Constraint Reformulation} \label{ssec:reform}

We take the complement of the joint chance constraint and employ Boole's inequality to convert the joint chance constraint into a affine combination of individual chance constraints,
\begin{equation}
    \pr{ \bigcup_{k=1}^N \bigcup_{i=1}^{q_k} \vec{G}_{ki} \bvec{x}(k) \geq h_{ki} }
    \leq  \sum_{k=1}^N \sum_{i=1}^{q_k} \pr{\vec{G}_{ki} \bvec{x}(k) \geq h_{ki}}
\end{equation}
We introduce risk allocation variables $\omega_{ki}$ to each of the individual chance constraints \cite{ono2008iterative}, and bound the sum of risk allocation variables,
\begin{subequations}\label{eq:quantile_reform_new_var}
\begin{align}
     \pr{\vec{G}_{ki} \bvec{x}(k) \geq h_{ki}} &\leq \omega_{ki}  \label{eq:quantile_reform_new_var_1_1}\\
     \sum_{k=1}^N \sum_{i=1}^{q_k} \omega_{ki} &\leq \alpha 
\end{align}
\end{subequations}
where $\omega_{ki} \in (0,1)$. 

Here, we need to find a suitable value of $\omega_{ki}$ such that we can solve the constraint. To this end, we impose an additional constraint to \eqref{eq:quantile_reform_new_var} based on the expectation and standard deviation of the random variable $\vec{G}_{ki}\bvec{x}(k)$,
\begin{equation} \label{eq:add_target}
    \ex{\vec{G}_{ki}\bvec{x}(k)} + \lambda_{ki}\std{\vec{G}_{ki} \bvec{x}(k)} \leq  h_{ki}
\end{equation}
with a non-negative optimization parameter $\lambda_{ki}$. Given Assumption \ref{assm:unimodal}, enforcement of \eqref{eq:add_target} guarantees that 
\begin{equation}\label{eq:first_bound}
\begin{split}
    &\pr{\vec{G}_{ki} \bvec{x}(k) \geq h_{ki} } \\
    & \ \leq \pr{\vec{G}_{ki} \bvec{x}(k) \geq \ex{\vec{G}_{ki}\bvec{x}(k)} + \lambda_{ki} \std{\vec{G}_{ki} \bvec{x}(k)}}  \\
     & \ \leq \frac{4}{9(\lambda_{ki}^2+1)}    
\end{split}
\end{equation}
by the one-sided Vysochanskij-Petunin inequality, so long as $\lambda_{ki} > \sqrt{5/3}$. Hence, we can choose
\begin{equation}
    \omega_{ki} = \frac{4}{9(\lambda_{ki}^2+1)}   
\end{equation}
and write \eqref{eq:quantile_reform_new_var}-\eqref{eq:first_bound} as 
\begin{subequations}\label{eq:quantile_reform_new_var_2}
\begin{align}
     \pr{\vec{G}_{ki} \bvec{x}(k) \geq h_{ki}}  &\leq  \frac{4}{9(\lambda_{ki}^2+1)}   \label{subeq:eq:quantile_reform_new_var_2_1} \\
     \ex{\vec{G}_{ki}\bvec{x}(k)} + \lambda_{ki} \std{\vec{G}_{ki} \bvec{x}(k)} & \leq  h_{ki} \label{subeq:eq:quantile_reform_new_var_2_2}  \\
     \sum_{k=1}^N \sum_{i=1}^{q_k}  \frac{4}{9(\lambda_{ki}^2+1)}   &\leq \alpha \label{subeq:eq:quantile_reform_new_var_2_3}\\
     \lambda_{ki} & > \sqrt{5/3} \label{eq:lambda_geq}
\end{align}
\end{subequations}
where \eqref{subeq:eq:quantile_reform_new_var_2_1} is a simplification of \eqref{eq:quantile_reform_new_var_1_1} and \eqref{eq:first_bound}, and \eqref{subeq:eq:quantile_reform_new_var_2_1}-\eqref{subeq:eq:quantile_reform_new_var_2_2} are iterated for all $i$ and $k$. 

Finally, we note that \eqref{subeq:eq:quantile_reform_new_var_2_1} serves only to act as an intermediary between \eqref{subeq:eq:quantile_reform_new_var_2_2} and \eqref{subeq:eq:quantile_reform_new_var_2_3}. Therefore, we can remove \eqref{subeq:eq:quantile_reform_new_var_2_1} as it is redundant. Further, as $\alpha<1/6$ implies $\lambda_{ki}  > \sqrt{5/3}$, \eqref{eq:lambda_geq} will always be met. Hence, \eqref{eq:quantile_reform_new_var_2} simplifies to 
\begin{subequations}\label{eq:quantile_reform_new_var_3}
\begin{align}
     \ex{\vec{G}_{ki}\bvec{x}(k)} + \lambda_{ki} \std{\vec{G}_{ki} \bvec{x}(k)} & \leq  h_{ki}  \label{eq:target_reform}\\
     \sum_{k=1}^N \sum_{i=1}^{q_k} \frac{4}{9(\lambda_{ki}^{2}+1)} &\leq \alpha \label{eq:target_lambda} 
\end{align}
\end{subequations}
for all $i$ and $k$. 

\begin{lem} \label{lem:1}
For the controller $\vec{U}$, if there exists risk allocation variables $\lambda_{ki}$ satisfying \eqref{eq:quantile_reform_new_var_3}, then $\vec{U}$ satisfies \eqref{prob:initial_eq_prob_constraints}.
\end{lem}

\begin{proof}
Given Assumption \ref{assm:unimodal}, the one-sided Vysochanskij-Petunin inequality guarantees that if \eqref{eq:target_reform} is satisfied, then \eqref{eq:first_bound} holds. Boole's inequality and De Morgan's law \cite{casella2002} guarantee that if \eqref{eq:target_lambda} holds then \eqref{prob:initial_eq_prob_constraints} is satisfied.
\end{proof}

We formally define the optimization problem that results from this reformulation.
\begin{subequations}\label{prob:big_prob_eq2}
    \begin{align}
        \underset{\vec{U}, \lambda_{ki}}{\mathrm{min}} \quad & J\left(
        \bvec{x}(1), \ldots, \bvec{x}(N),\vec{U}\right)  \\
        \mathrm{s.t. } \quad  & \vec{U} \in \mathcal{U}^N,  \\
        & \text{Expectation and standard deviation} \nonumber \\
        & \text{derived from dynamics } \eqref{eq:dynamics} \text{ with } \bar{x}(0) \label{eq:prob2_dyn}\\
        & \text{Constraint \eqref{eq:quantile_reform_new_var_3} for all $i$ and $k$} \label{eq:prob2_constraint} 
    \end{align}
\end{subequations}
\begin{reform} \label{prob:2}
    Under Assumptions \ref{assm:finite_moment}-\ref{assm:unimodal}, solve the stochastic optimization problem \eqref{prob:big_prob_eq2} with open loop control $\vec{U} \in  \mathcal{U}^N$, optimization parameters $\lambda_{ki}$, and probabilistic violation threshold $\alpha$.
\end{reform}

\begin{lem} \label{lem:2}
Any solution to Reformulation \ref{prob:2} is a conservative solution to Problem \ref{prob:1}.
\end{lem}
\begin{proof}
By Lemma \ref{lem:1}, \eqref{eq:prob2_dyn}-\eqref{eq:prob2_constraint} satisfy \eqref{prob:initial_eq_prob_constraints}. Here, \eqref{eq:prob2_dyn} replaces \eqref{eq:dynamics} as we only need the expectation and standard deviation derived from the dynamics. All other elements remain unchanged. Conservatism is introduced from Boole's inequality as equality is only achieved when constraints are independent. Similarly, the one-sided Vysochanskij-Petunin inequality never achieves equality and will introduce conservatism.
\end{proof}

Here, Reformulation \ref{prob:2} is a conservative but tractable reformulation of Problem \ref{prob:1}. While we cannot guarantee a solution exists to Reformulation \ref{prob:2}, we can guarantee that if a solution exists to Reformulation \ref{prob:2} it is also a solution to Problem \ref{prob:1}.

\begin{figure*}[t]
\normalsize
\newcounter{mytempeqncnt}
\setcounter{mytempeqncnt}{\value{equation}}
\setcounter{equation}{18}
\begin{equation} \label{eq:target_biconvex}
       \underbrace{\vec{G}_{ki} A^k \vec{x}(0)  + \vec{G}_{ki}\mathcal{A}(k) \ex{\boldsymbol{\mathcal{C}}}\vec{U}}_{\ex{\vec{G}_{ki} \bvec{x}(k)}} + \lambda_{ki} \underbrace{\left\| \var{\mathrm{vec}(\boldsymbol{\mathcal{C}})}^{1/2} (\vec{U} \otimes I_{nN} ) \mathcal{A}^{\top}(k) \vec{G}_{ki} ^{\top} \right\|}_{\std{\vec{G}_{ki} \bvec{x}(k)}} \leq  h_{ki} 
\end{equation}
\setcounter{equation}{\value{mytempeqncnt}}
\hrulefill
\end{figure*}

\subsection{Solving Reformulation \ref{prob:2}} \label{ssec:solve}

A biconvex problem has the following form:
\begin{subequations} \label{eq:biconvex}
\begin{align}
    \min_{x,y} & \; f(x,y) \\
    \mathrm{s.t.} & \; g_i(x,y) \leq 0 \quad \forall i \in \N 
\end{align}
\end{subequations}
where $x \in X \subseteq \R^n$, $y \in Y \subseteq \R^m$, $f(\cdot, \cdot): \R^n \times \R^m \rightarrow \R$ and $g_i(\cdot, \cdot): \R^n \times \R^m \rightarrow \R$ are convex when optimizing over one parameter while holding the other constant \cite{Gorski2007}. Next, we show that \eqref{eq:quantile_reform_new_var_3} elicits a biconvex constraint. We first find the closed form for \eqref{eq:quantile_reform_new_var_3}. The affine form of \eqref{eq:lin_dynamics} allows us to easily compute the expectation for an individual constraint via the linearity of the expectation operator, 
\begin{equation} \label{eq:exp_at_k}
\begin{split}
    \ex{\vec{G}_{ki}\bvec{x}(k)} = & \; \vec{G}_{ki} A^k \vec{x}(0)  + \vec{G}_{ki}\mathcal{A}(k) \ex{\boldsymbol{\mathcal{C}}}\vec{U}
\end{split}
\end{equation}
where $\ex{\boldsymbol{\mathcal{C}}}$ is the matrix consisting of the expectation of each element of $\boldsymbol{\mathcal{C}}$. Similarly, we can find the variance as
\begin{subequations} \label{eq:var_at_k}
\begin{align}
    & \var{\vec{G}_{ki} \bvec{x}(k)}  \\
    & \ = \left\| \var{\boldsymbol{\mathcal{C}}\vec{U}}^{1/2} \mathcal{A}^{\top}(k) \vec{G}_{ki} ^{\top} \right\|^2 \\
    & \ = \left\| \var{(\vec{U}^{\top} \otimes I_{nN} )\mathrm{vec}(\boldsymbol{\mathcal{C}})}^{1/2} \mathcal{A}^{\top}(k) \vec{G}_{ki} ^{\top} \right\|^2 \\
    & \ = \left\| \var{\mathrm{vec}(\boldsymbol{\mathcal{C}})}^{1/2} (\vec{U} \otimes I_{nN} ) \mathcal{A}^{\top}(k) \vec{G}_{ki} ^{\top} \right\|^2
\end{align}
\end{subequations}
Hence, we can write \eqref{eq:target_reform} as \eqref{eq:target_biconvex}
\stepcounter{equation}

As \eqref{eq:exp_at_k} is affine in the input, it is also convex. Further, \eqref{eq:var_at_k} allows us to write $\std{\vec{G}_{ki} \bvec{x}(k)}$ as a 2-norm and hence, it is also convex. As $\lambda_{ki}$ is a linear optimization parameter, $\lambda_{ki}\std{\vec{G}_{ki} \bvec{x}(k)}$ is biconvex. 

For completeness, we include Algorithm \ref{algo:acs} to demonstrate one method of solving \eqref{eq:target_biconvex} via the well-known alternate convex search method \cite{Leeuw1994}. While this method can only guarantee a locally optimal solution, Lemma \ref{lem:2} guarantees any solution will be a feasible solution to Problem \ref{prob:1}. 

\begin{algorithm}
  \caption{Computing solutions to \eqref{eq:biconvex} with alternate convex search}
	\label{algo:acs}
	\textbf{Input}: Feasible initial condition for $y$, denoted $y^\ast$, maximum number of iterations $n_{max}$.
	\\
	\textbf{Output}: Solution to \eqref{eq:biconvex}, $(x^\ast, y^\ast)$ 
	\begin{algorithmic}[1]
	\For{$i = 1 $ to $n_{max}$}
    \State Solve \eqref{eq:biconvex} assuming $y = y^\ast$; Set $x^\ast = x$
    \State Solve \eqref{eq:biconvex} assuming $x = x^\ast$; Set $y^\ast = y$
    \If{Solutions converged}
    \State \textbf{break}
    \EndIf
    \EndFor
  \end{algorithmic}
\end{algorithm}

\section{Results} \label{sec:results}

We demonstrate our method on a satellite rendezvous and docking problem with two different disturbances present in the control matrix. All computations were done on a 1.80GHz i7 processor with 16GB of RAM, using MATLAB, CVX \cite{cvx} and Gurobi \cite{gurobi}. All code is available at \url{https://github.com/unm-hscl/shawnpriore-random-control}. 

We consider the rendezvous of two satellites, referred to as the deputy and chief. The deputy spacecraft must remain in a predefined line-of-sight cone, and reach a target set that describes docking at the final time step. The relative dynamics are modeled via the Clohessy–Wiltshire equations \cite{wiesel1989_spaceflight}
\begin{subequations}
\begin{align}
\ddot x - 3 \omega^2 x - 2 \omega \dot y &= \frac{F_x}{m_c} \label{eq:cwh:a}\\
\ddot y + 2 \omega \dot x & = \frac{F_y}{m_c} \label{eq:cwh:b}\\
\ddot z + \omega^2 z & = \frac{F_z}{m_c}. \label{eq:cwh:c}
\end{align}   
\label{eq:cwh}
\end{subequations}
with input $\vec{u} = [ \begin{array}{ccc} F_x & F_y & F_z\end{array}]^\top$, orbital rate $\omega = \sqrt{\frac{\mu}{R^3_0}}$, gravitational constant $\mu$, orbital radius $R_0 = 42,164$km, and spacecraft mass $m_c=1$kg. We discretize \eqref{eq:cwh} under the assumption of impulse control with sampling time $60$s so that dynamics of the deputy are described by  
\begin{equation}
    \bvec{x}(k+1) = A \bvec{x}(k) + B \vec{u}(k) \label{eq:cwh_lin}
\end{equation}
with admissible input set $\mathcal{U} = [-0.1,0.1]^3$, and time horizon $N=5$, corresponding to 5 minutes of operation. 

The line-of-sight cone for time steps 1-4 is defined by
\begin{equation}
    G_k = \begin{bmatrix}
        -1 & 0 & 1 & 0 & 0 & 0 \\
        -1 & 1 & 0 & 0 & 0 & 0 \\
        -1 & 0 & -1 & 0 & 0 & 0 \\
        -1 & -1 & 0 & 0 & 0 & 0 \\
        1 & 0 & 0 & 0 & 0 & 0 
    \end{bmatrix} \; 
    \vec{h}_k = \begin{bmatrix}
        0 \\ 0 \\ 0 \\ 0 \\ 10
    \end{bmatrix}
\end{equation}
The terminal set is defined by 
\begin{equation}
    G_N = I_6 \otimes \begin{bmatrix}
        1 \\ -1
    \end{bmatrix} \; 
    \vec{h}_N = \begin{bmatrix}
        2 \\ 0 \\ 0.5 \cdot \vec{1}_{4} \\ 0.1 \cdot \vec{1}_{6}
    \end{bmatrix}
\end{equation}
We graphically represent the problem of interest in Figure \ref{fig:problem}. The probabilistic violation threshold $\alpha$ is set to 0.15 such that
\begin{equation}
        \pr{\bigcap_{k=1}^N  G_{k} \bvec{x}(k) \leq \vec{h}_{k} }  \geq  0.85
\end{equation}
The performance objective is based on fuel consumption,
\begin{equation}
    J\left(\bvec{x}(1), \ldots, \bvec{x}(N),\vec{U}\right) = \vec{U}^\top \vec{U}
\end{equation}
To solve this problem, we utilize the alternate convex search technique outlined in Algorithm \ref{algo:acs}.

\begin{figure}
    \centering
    \includegraphics[width=0.6\columnwidth]{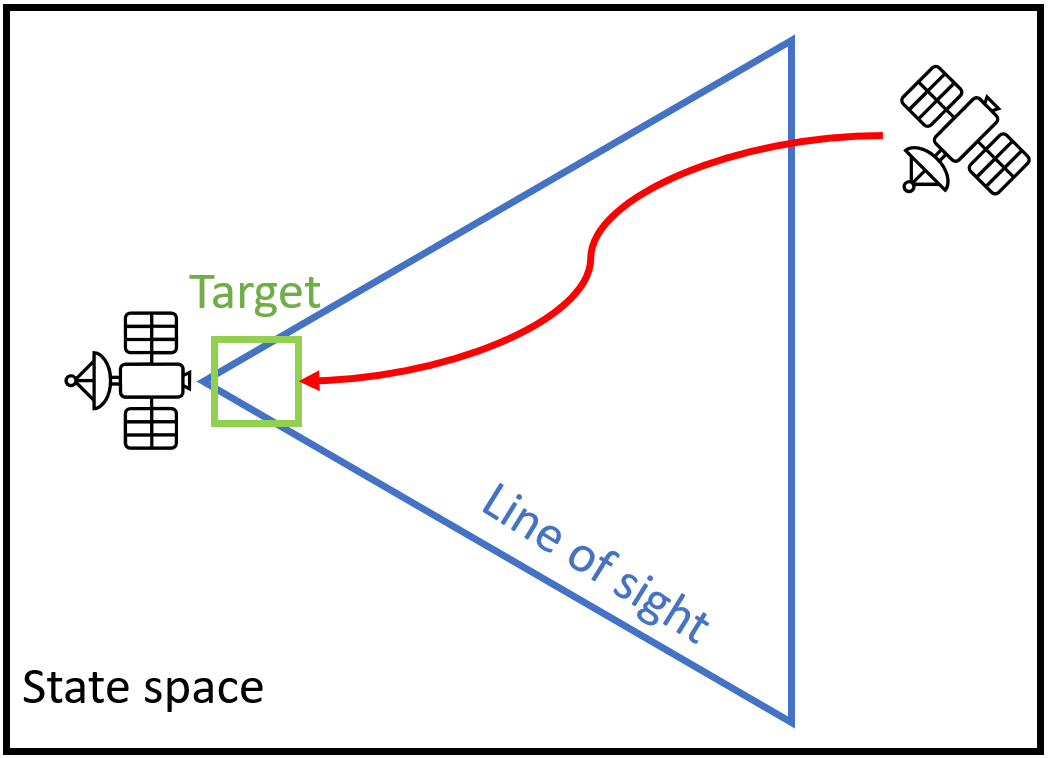}
    \caption{Graphic representation of the problem posed in Section \ref{sec:results}. Here, the dynamics of the deputy has stochasticity in the control matrix representing under-performing actuators and inaccuracies from impulse control assumptions. We attempt to find a control sequence that allows the deputy to rendezvous with the chief while meeting probabilistic time varying target set requirements. }
    \label{fig:problem}
\end{figure}

\subsection{Impulse Control Inaccuracies}

In this section, we consider the problem of inaccuracies caused by impulse control assumptions. In many engineering applications, the actuators will have to either ramp up or down to the desired system input. In either case, the vehicle will have received an incorrect level of input for some duration of the time step causing inaccuracies in the vehicles position. Here, we model the inaccuracies by multiplying each column by a random variable 
$$\Gamma_{ik} \overset{iid}{\sim} Gamma(10^{3},10^{-3})$$
Note that only the rows that correspond to positional elements of the state will be multiplied by a random variable. We write $\boldsymbol{B}(k)$ as
\begin{equation}
    \boldsymbol{B}(k) = \begin{bmatrix}
        59.9998 \Gamma_{1k} & 0.2619 \Gamma_{2k} & 0\\ 
        -0.2619 \Gamma_{1k} & 59.9992 \Gamma_{2k} & 0\\ 
        0 & 0 & 59.9998 \Gamma_{3k}\\ 
         1 & 0.0087 & 0\\ 
        -0.0087&  1& 0\\ 
        0 & 0 & 1
    \end{bmatrix}
\end{equation}
Here, $\Gamma_{ik}$ are independent. From the properties of the gamma distribution,
\begin{equation}
        \ex{\Gamma_{ik}} =  1 \quad
        \var{\Gamma_{ik}} =  10^{-3}
\end{equation}
which satisfies Assumption \ref{assm:finite_moment}. Since the $Gamma(10^{3},10^{-3})$ distribution is strong unimodal as per Definition \ref{defn:strong_unimodal}, Assumption  \ref{assm:unimodal} is also satisfied.

We compare the use of the one-sided Vysochanskij-Petunin inequality with the more broadly applicable Cantilli's inequality \cite{farina2016stochastic}. Using Cantilli's inequality does not require the constraints be unimodal and thus is a simpler method to use. However, for unimodal distributions, Cantilli's inequality is more conservative than the one-sided Vysochanskij-Petunin inequality. We see this in the resulting trajectory and solution cost, as presented in Figure \ref{fig:gamma} and Table \ref{tab:gamma}. The one-sided Vysochanskij-Petunin inequality resulted in a solution cost that was approximately $20\%$ less than that of the method using Cantelli's inequality. This difference in cost justifies the added burden of needing to verify constraint unimodality. 

To assess constraint satisfaction, we generated $10^5$ Monte Carlo sample disturbances for each approach. Table \ref{tab:gamma} shows that both methods satisfied the constraint for each sample taken. We expected both methods to be conservative, as neither the one-sided Vysochanskij-Petunin inequality nor Cantelli's inequality are tight bounds. We can also compare the relative conservativeness of the two methods in Figure \ref{fig:gamma}. As shown in the $X-Z$ plot on the right, the proposed method allows for the nominal trajectory to be closer to the edge of each hyperplane constraint. 

\begin{figure}
    \centering
    \includegraphics[width=0.9\columnwidth]{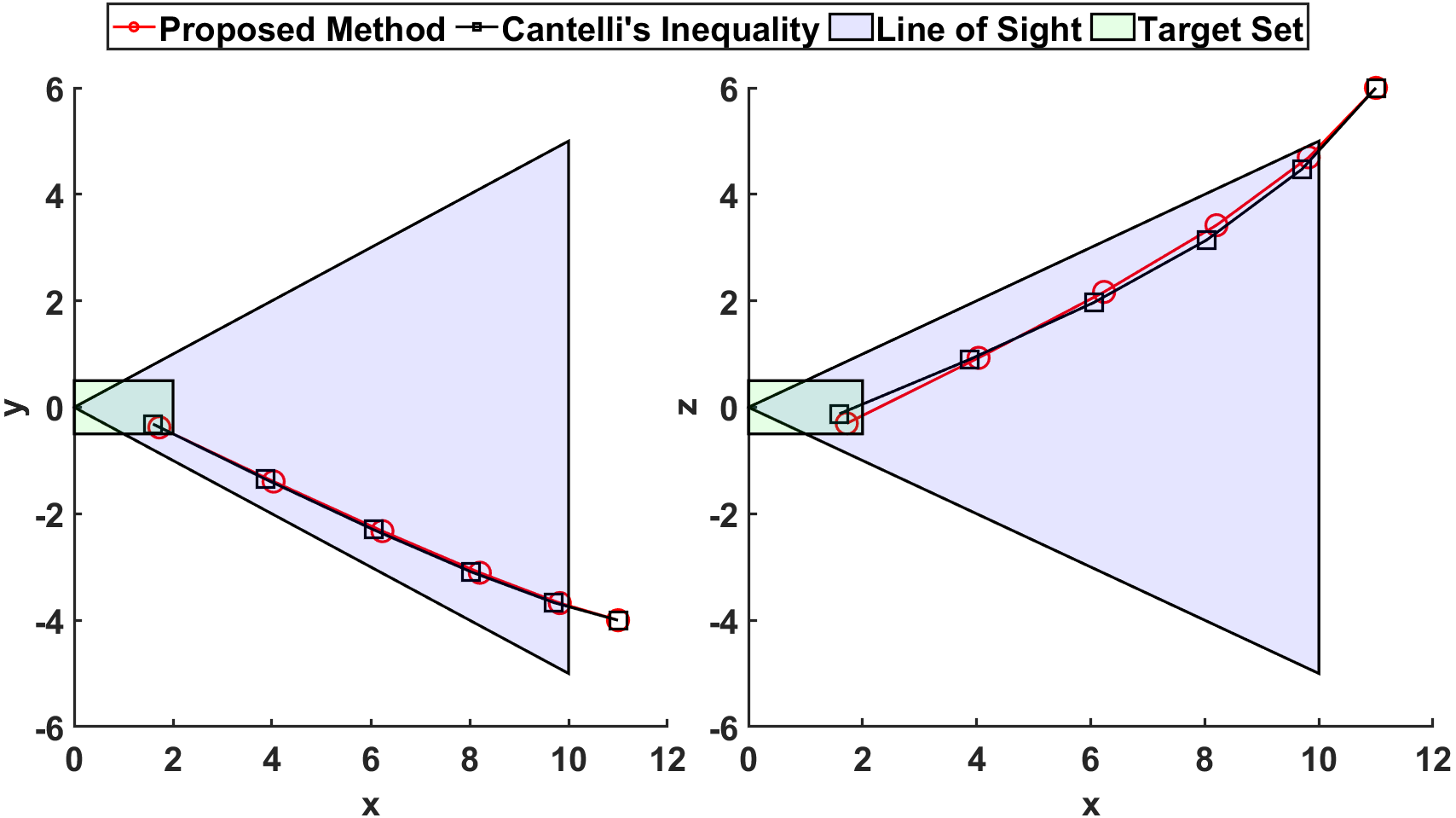}
    \caption{Comparison of expected trajectories between proposed method with the one-sided Vysochanskij–Petunin inequality (red line) and Cantelli's inequality (black line). Trajectories based on CWH dynamics with gamma distributed random elements in the control matrix.}
    \label{fig:gamma}
\end{figure}

\begin{table}
    \caption{Comparison of Solution and Computation Time for CWH Dynamics with Gamma Random Elements in the Control Matrix. Constraint Satisfaction was Measured by Proportion of $10^5$ Monte Carlo Samples that Satisfied the Constraint.}
    \centering
    \begin{tabular}{lcc}
         \hline \hline
         Metric &  Proposed Method & Cantelli's Inequality \cite{farina2016stochastic}\\
         \hline 
         Solve Time & 1.8740 s  & 2.1188  s\\ 
         Iterations &  2  & 2 \\ 
         Solution Cost & $1.030 \times 10^{-3}$ &   $1.282\times 10^{-3}$ \\ 
         Constraint Satisfaction  & 1.000 & 1.000 \\ \hline
    \end{tabular}
    \label{tab:gamma}
\end{table}

\subsection{Under-performing Actuators}

In this section, we consider the problem of under-performing actuators. In this scenario, we assume the impulse control assumption is valid, but the actuators do not reach the level of output as prescribed by the control input. Here, we model the under-performance by multiplying each column of the control matrix by a random variable 
$$\beta_{ik} \overset{iid}{\sim} Beta(152,8)$$
We write $\boldsymbol{B}(k)$ as
\begin{equation}
    \boldsymbol{B}(k) = \begin{bmatrix}
        59.9998 \beta_{1k} & 0.2619 \beta_{2k} & 0\\ 
        -0.2619 \beta_{1k} & 59.9992 \beta_{2k} & 0\\ 
        0 & 0 & 59.9998 \beta_{3k}\\ 
         \beta_{1k} & 0.0087 \beta_{2k}& 0\\ 
        -0.0087 \beta_{1k} &  \beta_{2k}& 0\\ 
        0 & 0 & \beta_{3k}
    \end{bmatrix}
\end{equation}
From the properties of the beta distribution, we know that
\begin{equation}
        \ex{\beta_{ik}} = 0.95 \quad
        \var{\beta_{ik}} =  2.95031 \times 10^{-4}
\end{equation}
and we can satisfy unimodality as the $Beta(152,8)$ distribution is strong unimodal as per Definition \ref{defn:strong_unimodal}. Hence, we know that both Assumptions \ref{assm:finite_moment} and \ref{assm:unimodal} have been satisfied.

We compare the proposed method with the scenario approach \cite{Campi2008}. To compute the number of samples needed to employ the scenario approach we use the formula $N_s \geq \frac{2}{\alpha}(\ln{\frac{1}{\delta}} + mN)$, where $\delta$ is a predefined confidence parameter. To allow for comparable results, we consider $\delta = 10^{-8}$, resulting in $N_s =446$ samples. We plot the expectation of the trajectories in Figure \ref{fig:beta}. We see that the two mean trajectories are nearly identical. 

Solution statistics and constraint satisfaction can be found in Table \ref{tab:beta}. To assess constraint satisfaction, we generated $10^5$ Monte Carlo sample disturbances for each approach. We see that while both methods satisfied the constraint, both were also conservative with respect to the safety threshold. We expected both methods to be conservative, as the one-sided Vysochanskij-Petunin inequality results in a loose bound and the scenario approach relies on samples without regard for the likelihood of the samples taken. However, while the solution cost and conservativeness of the two methods are similar, the solution was computed in almost half the time with the proposed method.

\begin{figure}
    \centering
    \includegraphics[width=0.9\columnwidth]{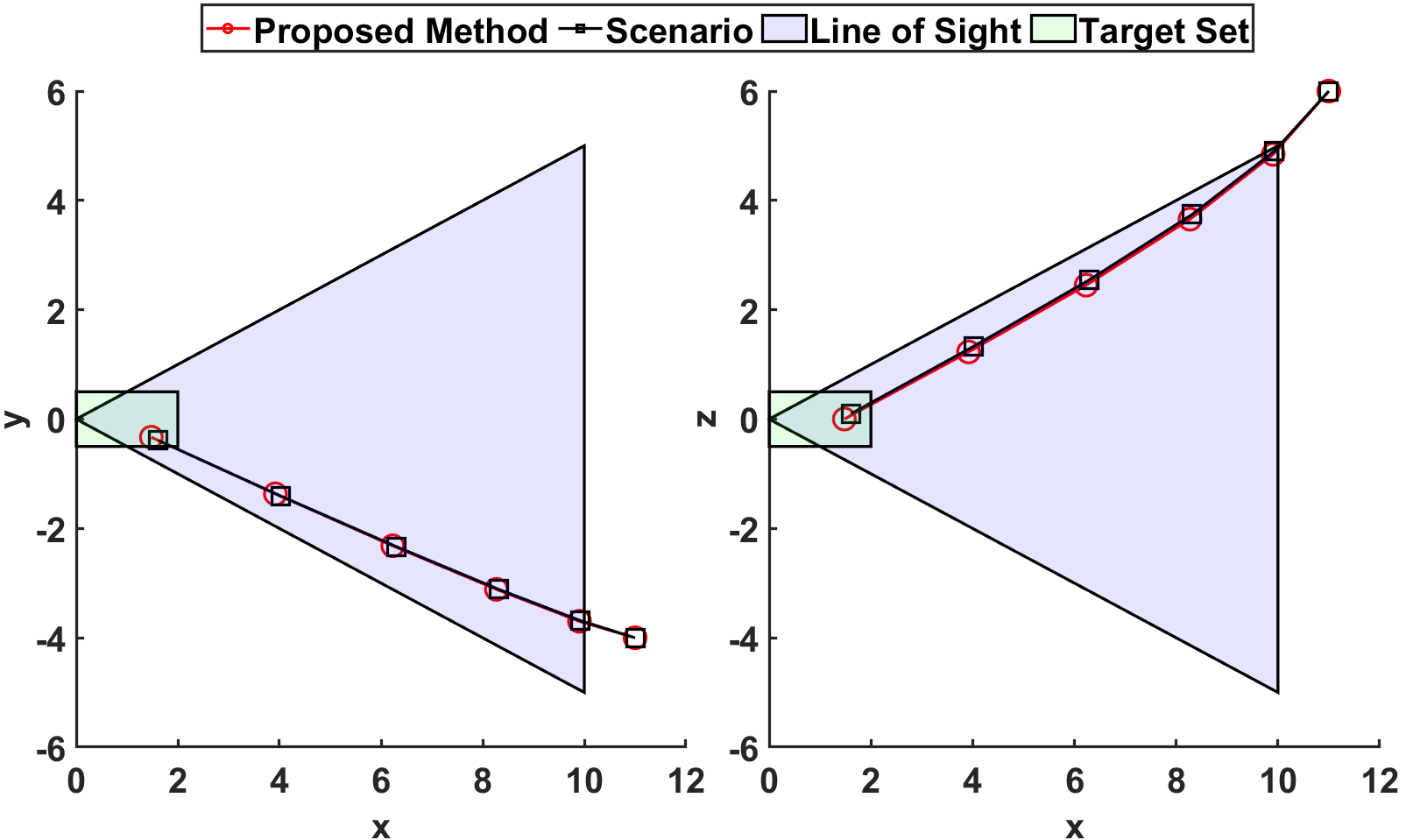}
    \caption{Comparison of expected trajectories between proposed method (red line) and scenario approach (black line). Trajectories based on CWH dynamics with beta distributed random elements in the control matrix. Notice the trajectories are nearly identical.}
    \label{fig:beta}
\end{figure}

\begin{table}
    \caption{Comparison of Solution and Computation Time for CWH Dynamics with Beta Random Elements in the Control Matrix. Constraint Satisfaction was Measured by Proportion of $10^5$ Monte Carlo Samples that Satisfied the Constraint.}
    \centering
    \begin{tabular}{lcc}
         \hline \hline
         Metric &  Proposed Method & Scenario Approach \\
         \hline 
         Solve Time & 22.2034 s  & 39.6245 s\\ 
         Iterations &  2  & N/A \\ 
         Solution Cost & $1.024 \times 10^{-3}$ &   $0.969\times 10^{-3}$ \\ 
         Constraint Satisfaction  & 0.9996 & 0.9893 \\ \hline
    \end{tabular}
    \label{tab:beta}
\end{table}

\section{Conclusions} \label{sec:conclusion}

We proposed a framework for solving stochastic optimal control problems for systems with random elements in the control matrix, subject to polytopic target set chance constraints. This framework is intended to help account for modeling inaccuracies in the control. Our approach relies on the one-sided Vysochanskij–Petunin inequality and Boole's inequality to reformulate the joint chance constraints into a series of individual biconvex constraints. We outlined the alternate convex search approach to solve the biconvex constraints. We demonstrated our method on two satellite rendezvous scenarios with inaccuracies resulting from impulse control assumptions and under-performing actuators. We compared our methods with the scenario approach and expectation methods reliant on Cantelli's inequality. We showed that our method resulted in a lower solution cost in comparison to Cantelli's inequality and shorter computation time in comparison to the scenario approach. 

\bibliography{main}

\begin{thebibliography}{10}

\bibitem{Branicky2000}
M.~Branicky, S.~Phillips, and W.~Zhang, ``Stability of networked control
  systems: explicit analysis of delay,'' in {\em Proceedings of the 2000
  American Control Conference. ACC (IEEE Cat. No.00CH36334)}, vol.~4,
  pp.~2352--2357 vol.4, 2000.

\bibitem{Huo2005}
Z.~Huo and H.~Fang, ``Robust $h_{\infty}$ filter design for networked control
  system with random time delays,'' in {\em 10th IEEE International Conference
  on Engineering of Complex Computer Systems (ICECCS'05)}, pp.~333--340, 2005.

\bibitem{Casavola2000}
A.~Casavola, M.~Giannelli, and E.~Mosca, ``Min–max predictive control
  strategies for input-saturated polytopic uncertain systems,'' {\em
  Automatica}, vol.~36, no.~1, pp.~125--133, 2000.

\bibitem{Calafiore2011}
G.~C. Calafiore and L.~Fagiano, ``Robust model predictive control via random
  convex programming,'' in {\em 2011 50th IEEE Conf. on Decision and Control
  and European Control Conf.}, pp.~1910--1915, 2011.

\bibitem{Gravell2021}
B.~Gravell and T.~Summers, ``Stochastic stability via robustness of linear
  systems,'' in {\em 2021 60th IEEE Conference on Decision and Control (CDC)},
  pp.~5918--5923, 2021.

\bibitem{Mahmoud1999}
M.~S. Mahmoud, ``Stability and $h_{\infty}$ filtering of linear
  parameter-varying discrete-time systems with state-delay,'' in {\em 1999
  European Control Conference (ECC)}, pp.~3709--3714, 1999.

\bibitem{Calafiore2012}
G.~C. Calafiore and L.~Fagiano, ``Model predictive control of stochastic lpv
  systems via random convex programs,'' in {\em 2012 IEEE 51st IEEE Conference
  on Decision and Control (CDC)}, pp.~3233--3238, 2012.

\bibitem{casella2002}
G.~Casella and R.~Berger, {\em Statistical Inference}.
\newblock Duxbury advanced series in statistics and decision sciences, Cengage
  Learning, 2002.

\bibitem{Mercadier2021}
M.~Mercadier and F.~Strobel, ``A one-sided vysochanskii-petunin inequality with
  financial applications,'' {\em European Journal of Operational Research},
  vol.~295, no.~1, pp.~374--377, 2021.

\bibitem{Leeuw1994}
J.~de~Leeuw, ``Block-relaxation algorithms in statistics,'' in {\em Information
  Systems and Data Analysis} (H.-H. Bock, W.~Lenski, and M.~M. Richter, eds.),
  (Berlin, Heidelberg), pp.~308--324, Springer, 1994.

\bibitem{Bertin1997}
E.~M.~J. Bertin, I.~Cuculescu, and R.~Theodorescu, {\em Strong unimodality},
  pp.~183--200.
\newblock Dordrecht: Springer Netherlands, 1997.

\bibitem{Ibragimov1956}
I.~A. Ibragimov, ``On the composition of unimodal distributions,'' {\em Theory
  of Prob. \& Its Applications}, vol.~1, no.~2, pp.~255--260, 1956.

\bibitem{Hartigan1985}
J.~A. Hartigan and P.~M. Hartigan, ``{The Dip Test of Unimodality},'' {\em The
  Annals of Statistics}, vol.~13, no.~1, pp.~70 -- 84, 1985.

\bibitem{Vysochanskij1980}
D.~F. Vysochanskij and Y.~I. Petunin, ``Justification of the 3{$\sigma$} rule
  for unimodal distributions,'' in {\em Theory of Probability and Mathematical
  Statistics}, vol.~21, pp.~25--36, 1980.

\bibitem{ono2008iterative}
M.~Ono and B.~Williams, ``Iterative risk allocation: A new approach to robust
  model predictive control with a joint chance constraint,'' in {\em IEEE Conf.
  Dec. \& Control}, pp.~3427--3432, 2008.

\bibitem{Gorski2007}
J.~Gorski, F.~Pfeuffer, and K.~Klamroth, ``Biconvex sets and optimization with
  biconvex functions: a survey and extensions,'' {\em Mathematical Methods of
  Operations Research}, vol.~66, pp.~373--407, Dec 2007.

\bibitem{cvx}
M.~Grant and S.~Boyd, ``{CVX}: Matlab software for disciplined convex
  programming, version 2.1.'' \url{http://cvxr.com/cvx}, Mar. 2014.

\bibitem{gurobi}
L.~Gurobi~Optimization, ``Gurobi optimizer reference manual,'' 2020.

\bibitem{wiesel1989_spaceflight}
W.~Wiesel, {\em Spaceflight Dynamics}.
\newblock New York: McGraw--Hill, 1989.

\bibitem{farina2016stochastic}
M.~Farina, L.~Giulioni, and R.~Scattolini, ``Stochastic linear model predictive
  control with chance constraints--a review,'' {\em J. Process Ctrl.}, vol.~44,
  pp.~53--67, 2016.

\bibitem{Campi2008}
M.~C. Campi, S.~Garatti, and M.~Prandini, ``The scenario approach for systems
  and control design,'' {\em IFAC Proceedings Volumes}, vol.~41, no.~2,
  pp.~381--389, 2008.
\newblock 17th IFAC World Congress.

\end{thebibliography}
\bibliographystyle{ieeetr}
\end{document}